\RequirePackage{fixltx2e}
\documentclass[english]{paper}
\usepackage{lmodern}
\usepackage[T1]{fontenc}
\usepackage[latin9]{inputenc}
\usepackage{geometry}
\usepackage{babel}
\usepackage{amsmath}
\usepackage{amsthm}
\usepackage{amssymb}
\usepackage{cancel}
\usepackage{stackrel}
\usepackage{setspace}
\usepackage{wasysym}
\usepackage[unicode=true,pdfusetitle,
 bookmarks=true,bookmarksnumbered=false,bookmarksopen=false,
 breaklinks=false,pdfborder={0 0 1},backref=false,colorlinks=false]
 {hyperref}

\makeatletter
\numberwithin{equation}{section}
\numberwithin{figure}{section}
\newcommand{\lyxaddress}[1]{
	\par {\raggedright #1
	\vspace{1.4em}
	\noindent\par}
}
\theoremstyle{plain}
\newtheorem{thm}{\protect\theoremname}
\theoremstyle{plain}
\newtheorem{lem}[thm]{\protect\lemmaname}
\theoremstyle{remark}
\newtheorem{rem}[thm]{\protect\remarkname}
\theoremstyle{plain}
\newtheorem{prop}[thm]{\protect\propositionname}
\theoremstyle{definition}
\newtheorem{defn}[thm]{\protect\definitionname}
\newenvironment{lyxlist}[1]
	{\begin{list}{}
		{\settowidth{\labelwidth}{#1}
		 \setlength{\leftmargin}{\labelwidth}
		 \addtolength{\leftmargin}{\labelsep}
		 }}
	{\end{list}}
\theoremstyle{plain}
\newtheorem{cor}[thm]{\protect\corollaryname}

\@ifundefined{date}{}{\date{}}
\makeatother

\providecommand{\corollaryname}{Corollary}
\providecommand{\definitionname}{Definition}
\providecommand{\lemmaname}{Lemma}
\providecommand{\propositionname}{Proposition}
\providecommand{\remarkname}{Remark}
\providecommand{\theoremname}{Theorem}

\begin{document}
\title{Characteristic-Dependent Linear Rank Inequalities via Complementary
Vector Spaces}
\author{Victor Peña\thanks{e-mail: \protect\href{mailto:vbpenam\%40unal.edu.co}{vbpenam@unal.edu.co}}
\& Humberto Sarria\thanks{e-mail: \protect\href{mailto:hsarriaz\%40unal.edu.co}{hsarriaz@unal.edu.co}}}
\institution{Departamento de Matemáticas, Universidad Nacional de Colombia, Bogotá,
Colombia}
\maketitle
\begin{abstract}
A characteristic-dependent linear rank inequality is a linear inequality
that holds by ranks of subspaces of a vector space over a finite field
of determined characteristic, and does not in general hold over other
characteristics. In this paper, we produce new characteristic-dependent
linear rank inequalities by an alternative technique to the usual
Dougherty's inverse function method \cite{9}. We take up some ideas
of Blasiak \cite{4}, applied to certain complementary vector spaces,
in order to produce them. Also, we present some applications to network
coding. In particular, for each finite or co-finite set of primes
$P$, we show that there exists a sequence of networks $\mathcal{N}\left(k\right)$
in which each member is linearly solvable over a field if and only
if the characteristic of the field is in $P$, and the linear capacity,
over fields whose characteristic is not in $P$, $\rightarrow0$ as
$k\rightarrow\infty$.
\end{abstract}

\lyxaddress{\textbf{\footnotesize{}Subject Classification: 68P30\vspace{-0.5cm}
}}

\section*{Introduction}

Network Coding is a branch of Information Theory introduced by Ahlswede,
Cai, Li and Yeung in 2000 that studies the problem of information
flow through a network \cite{1}. It has been proven that network
coding is a great tool for improving information management in contrast
to the usual way routing. It is known that an algorithm exists to
calculate the routing capacity of a network \cite{5} but it is unknown
if there is one for the linear capacity of a network, much less for
the non-linear capacity \cite{6}. Information inequalities play an
important role in the calculation of these capacities because upper
bounds have been found by treating the messages involved in the network
as random variables. So, any advance in the understanding of the regions
determined by entropies of random variables implies an advance in
network coding \cite{8,6,9,17}.

There are networks whose linear capacity is smaller than the non-linear
capacity \cite{6}. Therefore, in order to understand the linear capacity
of a network, it is necessary to study inequalities that are valid
for random variables induced by finite dimensional vector spaces.
It is well known that the entropy of these random variables is completely
determined by the dimension (usually referred to as rank) of the associated
vector spaces. The mentioned inequalities are called \emph{linear
rank inequalities}. Formally, a linear rank inequality is a linear
inequality that is always satisfied by ranks of subspaces of a vector
space. All information inequalities are linear rank inequalities but
not all linear rank inequalities are information inequalities \cite{key-3}.
The first example of a linear rank inequality that is not an information
inequality was found by Ingleton in \cite{13}. This inequality was
useful to calculate the linear capacity (over any field) of the Vámos
network \cite{7}. Other inequalities have been presented in \cite{8,10,19}. 

The linear capacity of a network depends on the characteristic of
the scalar field associated to the vector space of the network codes.
In other words, it is possible to achieve a higher rate of linear
communication by choosing one characteristic over another, an example
is the Fano network \cite{6,7}. Therefore, when we study linear capacities
over specific fields, it is also convenient to work with ``linear
rank inequalities'' that depend on the characteristic of the scalar
field associated to vector space. A \emph{characteristic-dependent
linear rank inequality} is a linear inequality that is always satisfied
by ranks of subspaces of a vector space over fields of certain characteristic
and does not in general hold over other characteristics. These are
the appropriate inequalities to calculate capacities over specific
fields. It is worth noting that all linear rank inequalities for up
to and including five variables are known and are all characteristic-independent
\cite{8}. The first two characteristic-dependent linear rank inequalities
(over seven variables) were presented by Blasiak, Kleinberg and Lubetzky
in 2011. Specifically, the first inequality holds for all fields whose
characteristic is not two and does not in general hold over characteristic
two. The second inequality holds for all fields whose characteristic
is two and does not in general hold over characteristics other than
two \cite{4}. Their application used linear program whose constraints
express information inequalities (and their inequalities) to produce
separation between linear and non-linear network coding. Using lexicographic
products, the separation is amplified, yielding a sequence of networks
in which the difference in linear and non-linear capacity is bigger
in each network.

In 2013, Dougherty, Freiling and Zeger presented two new characteristic-dependent
linear rank inequalities; again, one inequality is valid for characteristic
two and the other inequality is valid for every characteristic except
for two \cite{9}. The technique used to produce these inequalities
is called \emph{The inverse function method }and is different from
the technique used by Blasiak et al. in their inequalities. These
inequalities are then used to provide upper bounds for the linear
capacity of the Fano network and non-Fano network. In 2014, E. Freiling
in \cite[Ph.D. thesis]{11}, for each finite or co-finite set of prime
numbers, constructed a characteristic-dependent linear rank inequality
that is valid only for vector spaces over fields whose characteristic
is in the aforementioned set. The technique that Freiling used is
a generalization of the inverse function method. He also showed that
for each finite or co-finite set of primes, there exists a network
that is linearly solvable over a field if and only if the characteristic
of the field is in the set. In this thesis appears the natural question:
Are there other techniques to tighten these inequalities?

\textbf{Organization of the work and contributions.} This work is
organized into two sections. In section 1, we introduce the basic
definitions related to Linear Algebra and Information Theory. Then,
we produce new characteristic-dependent linear rank inequalities by
taking the central ideas of Blasiak et al. \cite{4} but modifying
some of their arguments: We take a matrix which is a generalization
matrix of the representation matrix of the Fano and non-Fano matroids.
Some vectorial matroids associated to this matrix are known in \cite{14}.
This matrix is used as a guide to extract some properties of vector
spaces and obtain certain conditional inequalities. Then, we turn
these inequalities into characteristic-dependent linear rank inequalities.
We also present some cases when the desired inequalities are indeed
true over any field. In section 2, we review some concepts of Network
Coding and Index Coding, as well as some results of Blasiak \cite{4}
in order to define our linear programs which are useful for our application
theorem to network coding: For each finite or co-finite set of primes
$P$, we show that there exists a sequence of networks $\mathcal{N}\left(k\right)$
in which each member is linearly solvable over a field if and only
if the characteristic of the field is in $P$, and the linear capacity,
over fields whose characteristic is not in $P$, $\rightarrow0$ as
$k\rightarrow\infty$. This means that we have a sequence of solvable
networks in which we can achieve a higher rate of linear communication
by choosing one characteristic in $P$ over another in the complement
set of $P$, and the rate of linear communication on this last set
can be as bad as we want. We remark that these networks are associated
to index coding instances from vector matroids whose matrix is used
in section 1. Also, we remark that the gap in capacities is obtained
via lexicograph product and improves the above mentioned result of
Freiling \cite[Theorem 3.3.1 and 3.3.2]{11}. Additionally, as a corollary
we present many sequences of networks which the rate of (non-linear)
communication is better than the rate of linear communication. It
is notable that one of these sequences is a modificated version of
the sequence that was presented by Blasiak et al. \cite[Theorem 1.2]{4}.
By last, we show that our sequences of networks have a good coding
gain \cite{12}.

\section{Characteristic-dependent linear rank inequalities}

Let $A$, $A_{1}$, $\ldots$, $A_{n}$, $B$ be vector subspaces
of a finite dimensional vector space $V$. There is a correspondence
between linear rank inequalities and information inequalities associated
to certain class of random variables induced by vector spaces, see
\cite[Theorem 2]{key-3}. So, we can use notation of information theory
to refer dimension of vector spaces. Let $A_{I}:=\underset{i\in I}{\sum}A_{i}$
denote the span or sum of $A_{i}$, $i\in I\subseteq\left[n\right]:=\left\{ 1,2,\ldots,n\right\} $,
the entropy of $A_{I}$ is the dimension, $\text{H}\left(A_{I}\right)=\dim\left(A_{i},i\in I\right)$.
The mutual information of $A$ and $B$ is $\text{I}\left(A;B\right)=\dim\left(A\cap B\right)$.
If $B$ is a subspace of a subspace $A$, then we denote the \emph{codimension}
of $B$ in $A$ by $\text{codim}_{A}\left(B\right):=\text{H}\left(A\right)-\text{H}\left(B\right)$.
For $A$ and $B$ vector subspaces, $\text{H}\left(A\mid B\right)=\text{codim}_{A}\left(A\cap B\right)$. 

The sum $A+B$ is a direct sum if and only if $A\cap B=O$, the notation
for such a sum is $A\oplus B$. Subspaces $A_{1}$, ..., $A_{n}$
are called \emph{mutually complementary} subspaces in $V$ if every
vector of $V$ has an unique representation as a sum of elements of
$A_{1}$, ..., $A_{n}$. Equivalently, they are mutually complementary
subspaces in $V$ if and only if $V=A_{1}\oplus\cdots\oplus A_{n}$.
In this case, $\pi_{S}$ denotes the canonical projection function
$V\twoheadrightarrow\underset{i\in S}{\bigoplus}A_{i}$.

In the principal proof of this section we will need to calculate the
difference in dimension between vector spaces, so inequalities associated
to codimension given by the following two lemmas are important. 
\begin{lem}
\label{lemma basico 1desigualdad condicional rango lineal basica general}For
any subspaces $A_{1},\ldots,A_{m},A_{1}',\ldots,A_{m}'$ of finite
dimensional vector space $V$ such that $A_{i}'\leq A_{i}$,
\[
\mathrm{codim}_{A_{\left[m\right]}}A_{\left[m\right]}'\leq\stackrel[i=1]{m}{\sum}\mathrm{codim}_{A_{i}}A_{i}'
\]
with equality if and only if $A_{k+1}\cap A_{\left[k\right]}=A_{k+1}'\cap A_{\left[k\right]}'$
for all $k$.
\end{lem}

\smallskip{}

\begin{lem}
\label{lema basico 2 }For any subspaces $A$, $B$, $C$ of finite
dimensional vector space $V$ such that $B\leq A$, 
\[
\mathrm{codim}_{\left(A\cap C\right)}\left(B\cap C\right)\leq\mathrm{codim}_{A}B
\]
with equality if and only if there exists a subspace of $C$ which
is complementary to $B$ in $A$.
\end{lem}

\textbf{Inequalities using a suitable matrix as a guide.} For $n\geq2$,
$L_{n}$ denotes the $\left(n+1\right)\times\left(2n+3\right)$-matrix
\[
\begin{array}{c}
A_{1}\,\cdots\,A_{n}\,A_{n+1}B_{1}\,\cdots\,B_{n}\,B_{n+1}\,C\\
\left(\begin{array}{ccccccccc}
1 & \cdots & 0 & 0 & 0 & \cdots & 1 & 1 & 1\\
0 & \cdots & 0 & 0 & 1 & \cdots & 1 & 1 & 1\\
\vdots & \cdot & \vdots & \vdots & \vdots & \cdot & \vdots & \vdots & \vdots\\
0 & \cdots & 1 & 0 & 1 & \cdots & 0 & 1 & 1\\
0 & \cdots & 0 & 1 & 1 & \cdots & 1 & 0 & 1
\end{array}\right)
\end{array}.
\]
The rank of the submatrix $B_{\left[n+1\right]}$ depends on the field
where its inputs are defined: If the characteristic of the field divides
$n$, the rank is $n$; and if the characteristic of the field does
not divide $n$, the rank is maximum. Lemmas \ref{lema desigualdad condicional general caracteristica que divide n}
and \ref{lema desigualdad condicional sobre caracteristica que no dividen n}
(with the help of Lemma \ref{lemma 1:  distingue caracteristicas que dividen n y la que no lo divide})
present a general version of this. Specifically, these lemmas abstract
the properties of linear independence between the vector spaces (over
a field with certain characteristic) generated by the columns of $L_{n}$
to obtain inequalities associated to the rank of the vector space
generates by the columns of the submatrix $B_{\left[n+1\right]}$
and the rank of the vector space generates by the column $C$.
\begin{lem}
\label{lemma 1:  distingue caracteristicas que dividen n y la que no lo divide}Let
$A_{1}$, $A_{2}$, $\ldots$, $A_{n+1}$ be mutually complementary
vector subspaces of a vector space $V$ over a field $\mathbb{F}$,
and $C$ a subspace of $V$ such that the sum of $\stackrel[i=1,i\neq k]{n+1}{\bigoplus}A_{i}$
and $C$ is a direct sum for all $k$. Then
\[
\mathrm{H}\left(\left\{ \pi_{\left[n+1\right]-i}\left(C\right)\right\} _{i=1}^{n+1}\right)=\left\{ \begin{array}{c}
n\mathrm{H}\left(C\right)\,\,\,\,\,\,\,\,\,\,\,\,\,\text{ if }\text{char}\left(\mathbb{F}\right)\mid n\\
\left(n+1\right)\mathrm{H}\left(C\right)\text{ if }\text{char}\left(\mathbb{F}\right)\nmid n\text{.}
\end{array}\right.
\]
\end{lem}

\begin{proof}
\emph{We have the following claim:} A non-zero element of $C$ has
$n+1$ non-zero coordinates. Moreover, for all $i$, $\text{H}\left(\pi_{\left[n+1\right]-i}\left(C\right)\right)=\text{H}\left(C\right)$.\emph{
Proof of claim.} Let $v\in C$, we can write $v$ as $\stackrel[i=1]{n+1}{\sum}v_{i}$,
where $v_{i}\in A_{i}$ for $i=1,\ldots,n+1$. If $v_{k}=0$ for some
$1\leq k\leq n+1$, then $v\in\stackrel[i=1,i\neq k]{n+1}{\bigoplus}A_{i}$
but $C$ is complementary to this space. It follows $v=O$.$\Square$

Now, we consider the case when $\text{char}\left(\mathbb{F}\right)$
divides $n$. For any $v=\stackrel[i=1]{n+1}{\sum}v_{i}\in V$, taking
into account that $n=0$ and $n-1$ is invertible in $\mathbb{F}$,
we get 
\[
\frac{1}{n-1}\stackrel[i=1]{n}{\sum}\pi_{\left[n+1\right]-i}\left(v\right)=\frac{1}{n-1}\stackrel[i=1]{n}{\sum}\stackrel[j=1,j\neq i]{n+1}{\sum}v_{j}\,\,\,\,\,\,\,\,\,\,\,\,\,\,\,\,\,\,\,\,\,\,\,\,\,\,\,\,\,\,\,\,\,\,\,\,
\]
\[
\,\,\,\,\,\,\,\,\,\,\,\,\,\,\,\,\,\,\,\,\,\,\,\,\,\,\,\,\,\,\,\,\,\,\,\,\,\,\,\,\,\,\,\,\,\,\,\,\,\,\,\,\,\,\,\,=\frac{1}{n-1}\stackrel[i=1]{n}{\sum}\stackrel[j=1,j\neq i]{n}{\sum}v_{j}+\cancelto{O}{\frac{n}{n-1}v_{n+1}}
\]
\[
=\stackrel[i=1]{n}{\sum}v_{i}\,\,\,\,\,\,\,\,\,\,\,\,\,\,\,\,\,\,\,
\]
\[
\,\,\,\,\,\,\,\,\,\,\,\,\,\,\,=\pi_{\left[n+1\right]-\left(n+1\right)}\left(v\right)\text{.}
\]
Hence, $\pi_{\left[n+1\right]-\left(n+1\right)}\left(C\right)\leq\stackrel[i=1]{n}{\sum}\pi_{\left[n+1\right]-i}\left(C\right)$.
Furthermore, the subspaces $\pi_{\left[n+1\right]-i}\left(C\right)$
with $i\in\left[n\right]$ form a direct sum. In effect, let $v_{i}=\stackrel[j=1]{n+1}{\sum}v_{i}^{j}\in C$,
$i\in\left[n\right]$ such that $\stackrel[i=1]{n}{\sum}\pi_{\left[n+1\right]-i}\left(v_{i}\right)=O\text{.}$
Then for every $1\leq k\leq n$, we get $\stackrel[i=1,i\neq k]{n}{\sum}v_{i}^{k}$
and $\stackrel[i=1]{n}{\sum}v_{i}^{n+1}$ are equal to zero. Then,
applying claim to each $1\leq k\leq n$, we get $\stackrel[i=1,i\neq k]{n}{\sum}v_{i}$
and $\stackrel[i=1]{n}{\sum}v_{i}$ are equal to zero vector. Thus,
for every $1\leq k\leq n$, we get $v_{k}=\stackrel[i=1]{n}{\sum}v_{i}-\stackrel[i=1,i\neq k]{n}{\sum}v_{i}=O$.
Consequently, the subspaces $\pi_{\left[n+1\right]-i}\left(C\right)$,
with $i\in\left[n\right]$, are mutually complementary. Applying claim
to this fact, we get $\text{H}\left(\pi_{\left[n+1\right]-i}\left(C\right),i\in\left[n+1\right]\right)=\text{H}\left(\pi_{\left[n+1\right]-i}\left(C\right),i\in\left[n\right]\right)=n\text{H}\left(C\right)$.
Now, we consider the case when $\text{char}\left(\mathbb{F}\right)$
does not divide $n$. It is enough to prove that $\stackrel[i=1]{n+1}{\sum}\pi_{\left[n+1\right]-i}\left(C\right)$
is a direct sum. In effect, for each $i=1$, $\ldots$, $n+1$ take
$v_{i}=\stackrel[j=1]{n+1}{\sum}v_{i}^{j}$ in $C$ such that $\stackrel[i=1]{n+1}{\sum}\pi_{\left[n+1\right]-i}\left(v_{i}\right)=O$.
Then for every $1\leq k\leq n+1$, we get $\stackrel[i=1,i\neq k]{n+1}{\sum}v_{i}^{k}=0$,
which for claim implies, $\stackrel[i=1,i\neq k]{n+1}{\sum}v_{i}=O$
for all $k$. Fixed $j$, add member to member all these inequalities
except the inequality corresponding to $k=j$, we get
\[
O=\stackrel[k=1,k\neq j]{n+1}{\sum}\left(\stackrel[\begin{array}{c}
i=1,i\neq k\end{array}]{n+1}{\sum}v_{i}\right).
\]
\[
=nv_{j}+\cancelto{O}{\left(n-1\right)\stackrel[i=1,i\neq j]{n+1}{\sum}v_{i}}.
\]
Since $\text{char}\left(\mathbb{F}\right)$ does not divide $n$,
$v_{j}=O$.
\end{proof}
\begin{rem}
We remark that a subspace $C$ as described in previous lemma holds
$\mathrm{H}\left(C\right)\leq1$.
\end{rem}

\begin{lem}
\label{lema desigualdad condicional general caracteristica que divide n}
Let $A_{1}$, $A_{2}$, $\ldots$, $A_{n+1}$, $B_{1}$, $B_{2}$,
$\ldots$, $B_{n+1}$, $C$ be subspaces of a finite-dimensional vector
space $V$ over a scalar field $\mathbb{F}$ whose field characteristic
divides $n$ and

(i) $A_{1}$,$\ldots$, $A_{n+1}$ are mutually complementary in $V$,
and subspaces $C$ and $A_{\left[n+1\right]-k}$ form a direct sum
for all $k$. 

(ii) $B_{k}\leq A_{\left[n+1\right]-k}\cap\left(A_{k}+C\right)$ for
all $k$.

Then $\mathrm{H}\left(B_{\left[n+1\right]}\right)\leq n\mathrm{H}\left(C\right)$.
\end{lem}

\begin{proof}
By hypotheses (i) and condition of the characteristic, we apply lemma
\ref{lemma 1:  distingue caracteristicas que dividen n y la que no lo divide}
to get
\begin{equation}
\text{H}\left(\pi_{\left[n+1\right]-i}\left(C\right),i\in\left[n+1\right]\right)=n\text{H}\left(C\right)\text{.}\label{eq:aplicacion de lema 1 a lema2}
\end{equation}
Furthermore, $\pi_{\left[n+1\right]-k}\left(C\right)=\left(C+A_{k}\right)\cap A_{\left[n+1\right]-k}$,
for all $k$. In effect, let $v\in C$ such that $v=\stackrel[i=1]{n+1}{\sum}v_{i}$,
where $v_{i}\in A_{i}$, $i\in\left[n+1\right]$ and fixed $k\in\left[n+1\right]$.
Noting that $\pi_{\left[n+1\right]-k}\left(v\right)=\stackrel[i=1,i\neq k]{n+1}{\sum}v_{i}$$=v-v_{k}$,
we get $\pi_{\left[n+1\right]-k}\left(v\right)\in\left(C+A_{k}\right)\cap A_{\left[n+1\right]-k}$.
To prove the other contention, let $u\in\left(C+A_{k}\right)\cap A_{\left[n+1\right]-k}$.
Then, there exist $v\in C$ and $v_{i}\in A_{i}$, for each $i\in\left[n+1\right]$
such that $u=v-v_{k}=\stackrel[i=1,i\neq k]{n+1}{\sum}v_{i}$. Thus
$v=\stackrel[i=1]{n+1}{\sum}v_{i}$ and $u=\pi_{\left[n+1\right]-k}\left(v\right)\in\pi_{\left[n+1\right]-k}\left(C\right)$.
Thus, the desired equality is true. Therefore, using hypothesis (ii),
we have that $B_{k}\leq\pi_{\left[n+1\right]-k}\left(C\right)$ which
implies $\stackrel[k=1]{n+1}{\sum}B_{k}\leq\stackrel[k=1]{n+1}{\sum}\pi_{\left[n+1\right]-k}\left(C\right)$.
From this and equation (\ref{eq:aplicacion de lema 1 a lema2}), we
get $\text{H}\left(B_{\left[n+1\right]}\right)\leq n\text{H}\left(C\right).$
\end{proof}
\begin{lem}
\label{lema desigualdad condicional sobre caracteristica que no dividen n}Let
$A_{1}$, $A_{2}$, $\ldots$, $A_{n+1}$, $B_{1}$, $B_{2}$, $\ldots$,
$B_{n+1}$, $C$ be subspaces of a finite-dimensional vector space
$V$ over a scalar field $\mathbb{F}$ whose field characteristic
does not divide $n$ and

(i) $A_{1}$, $\ldots$, $A_{n+1}$ are mutually complementary in
$V$, and subspaces $C$ and $A_{\left[n+1\right]-k}$ form a direct
sum for all $k$.

(ii) $B_{k}\leq A_{\left[n+1\right]-k}$ for all $k$.

(iii) $C\leq A_{k}+B_{k}$ for all $k$.

Then $\left(n+1\right)\mathrm{H}\left(C\right)\leq\mathrm{H}\left(B_{\left[n+1\right]}\right)$.
\end{lem}

\begin{proof}
By hypotheses (i) and condition of the characteristic we apply lemma
\ref{lemma 1:  distingue caracteristicas que dividen n y la que no lo divide}
to get
\begin{equation}
\text{H}\left(\pi_{\left[n+1\right]-i}\left(C\right):i\in\left[n+1\right]\right)=\left(n+1\right)\text{H}\left(C\right).\label{eq:lema para  desigualdades sobre los que no dividen a n}
\end{equation}
Furthermore, $\pi_{\left[n+1\right]-k}\left(C\right)\leq B_{k}$ for
all $k$. In effect, fixed $k\in\left[n+1\right]$ and let $v=\stackrel[i=1]{n+1}{\sum}v_{i}\in C$,
where $v_{i}\in A_{i}$. By hypothesis (iii), there exist $a_{k}\in A_{k}$
and $b_{k}\in B_{k}$ such that $v=a_{k}+b_{k}$. By hypothesis (ii),
there exist $a_{j}\in A_{j}$, for $j\in\left[n+1\right]-k$, such
that $b_{k}=\stackrel[j=1,j\neq k]{n+1}{\sum}a_{j}$. Then $v=\stackrel[i=1]{n+1}{\sum}v_{i}=a_{k}+\stackrel[j=1,j\neq k]{n+1}{\sum}a_{j}$
but $v$ has unique writing in terms of $A_{i}$, $i\in\left[n+1\right]$,
in particular, $a_{k}=v_{k}$. We get $\pi_{\left[n+1\right]-k}\left(v\right)=v-v_{k}=b_{k}\in B_{k}$.
In other words, $\pi_{\left[n+1\right]-k}\left(C\right)\leq B_{k}$.
Hence, $\stackrel[k=1]{n+1}{\sum}\pi_{\left[n+1\right]-k}\left(C\right)\leq\stackrel[k=1]{n+1}{\sum}B_{k}$.
Therefore, using equation (\ref{eq:lema para  desigualdades sobre los que no dividen a n})
we get, $\text{H}\left(B_{\left[n+1\right]}\right)\geq\left(n+1\right)\text{H}\left(C\right)\text{.}$
\end{proof}
Inequalities imply by lemmas \ref{lema desigualdad condicional general caracteristica que divide n}
and \ref{lema desigualdad condicional sobre caracteristica que no dividen n}
are conditional characteristic-dependent linear rank inequalities,
in the sense that they are true only for vector spaces with certain
relations of linear dependency. Theorems \ref{theorema fuerte caracterisitica rango lineal para conjunto finito}
and \ref{teorema fuerte desigualdad caracteristica rango lineal para cofinitos}
will use these inequalities to obtain characteristic-dependent linear
rank inequalities. The demonstrations consists of finding vector subspaces
of the original vector subspaces that satisfy the conditions of these
lemmas. Then, we find an upper bounds and a lower bounds of the inequalities
imply by these lemmas in terms of the original subspaces. To accomplish
this, we introduce the following construction: First, we build mutually
complementary subspaces $A_{1}'$,..., $A_{n+1}'$ in $A_{\left[n+1\right]}$
from $A_{1},...,A_{n+1}$: Define $A_{1}':=A_{1}$, and for $k=2$,
$\ldots$, $n+1$ denote by $A_{k}'$ a subspace of $A_{k}$ which
is a complementary subspace to $A_{\left[k-1\right]}$ in $A_{\left[k\right]}$.
Then $A_{1}'$, ..., $A_{n+1}'$ are mutually complementary and the
following equations hold:
\begin{equation}
\text{codim}_{A_{k}}\left(A_{k}'\right)=\text{I}\left(A_{\left[k-1\right]};A_{k}\right),\label{ecuaciones de codimension de los Ai primas}
\end{equation}
where $A_{0}=O$. Second, we built a subspace $\bar{C}$ of $C\cap A_{\left[n+1\right]}'$
such that $\bar{C}$ and $A_{\left[n+1\right]-k}'$ form a direct
sum for all $k$. Let $C^{\left(0\right)}:=C\cap A_{\left[n+1\right]}$.
Recursively, for $k=1$, $\ldots$, $n+1$ denote by  $C^{\left(k\right)}$
a subspace of $C^{\left(k-1\right)}$ which is a complementary subspace
to $A_{\left[n+1\right]-k}'$ in $C^{\left(k-1\right)}+A_{\left[n+1\right]-k}'$.
We denote $\bar{C}:=C^{\left(n+1\right)}$, this space satisfies the
required condition and the following equation:
\begin{equation}
\text{codim}_{C}\left(\bar{C}\right)\leq\text{H}\left(C\mid A_{\left[n+1\right]}\right)+\stackrel[i=1]{n+1}{\sum}\text{I}\left(A_{\left[n+1\right]-i};C\right)\text{.}\label{eq: codimension de C (raya horizontal arriba)}
\end{equation}
Summarizing, from $V$, $A_{1}$, $\ldots$, $A_{n+1}$ and $C$,
we built a tuple of vector subspaces
\begin{equation}
A_{1}',\ldots,A_{n+1}',\bar{C}\label{eq: construcci=0000F3n de espacios que son sumas directas}
\end{equation}
in which the sum of any members is a direct sum. We remark that this
tuple is not unique but in the proofs of the following two theorems
we will fix one of these.
\begin{thm}
\label{theorema fuerte caracterisitica rango lineal para conjunto finito}For
any $n\geq2$. Let $A_{1}$, $A_{2}$, $\ldots$, $A_{n+1}$, $B_{1}$,
$B_{2}$, $\ldots$, $B_{n+1}$, $C$ be subspaces of a finite-dimensional
vector space $V$ over a scalar field $\mathbb{F}$ whose field characteristic
divides $n$,
\[
\text{H}\left(B_{\left[n+1\right]}\right)\leq n\text{I}\left(A_{\left[n+1\right]};C\right)+\stackrel[i=1]{n+1}{\sum}\text{H}\left(B_{i}\mid A_{\left[n+1\right]-i}\right)+\stackrel[i=1]{n+1}{\sum}\text{H}\left(B_{i}\mid A_{i},C\right)+n\stackrel[i=2]{n}{\sum}\text{I}\left(A_{\left[i-1\right]};A_{i}\right)
\]
\[
+\left(n+1\right)\left[\text{I}\left(A_{\left[n\right]};A_{n+1}\right)+\text{H}\left(C\mid A_{\left[n+1\right]}\right)+\stackrel[i=1]{n+1}{\sum}\text{I}\left(A_{\left[n+1\right]-i};C\right)\right]\text{.}
\]
\end{thm}

\begin{proof}
The tuple (\ref{eq: construcci=0000F3n de espacios que son sumas directas})
obtained from the given vector spaces satisfies the condition (i)
of the lemma in the space $V'=A_{\left[n+1\right]}$. To meet condition
(ii), we define for \textbf{$k=1$ }to\textbf{ $k=n+1$, }$B_{k}':=B_{k}\cap\left(A_{\left[n+1\right]-k}'\right)\cap\left(A_{k}'+\bar{C}\right)$.
Subspaces $A_{1}'$, $...$, $A_{n+1}'$, $B_{1}'$, $...$, $B_{n+1}'$,
$\bar{C}$ of $V'$ satisfy all hypothesis of lemma \ref{lema desigualdad condicional general caracteristica que divide n}
over a scalar field $\mathbb{F}$ whose field characteristic divides
$n$, we get
\begin{equation}
\text{H}\left(B_{\left[n+1\right]}'\right)\leq n\text{H}\left(\bar{C}\right)\text{.}\label{eq:-17}
\end{equation}
An upper bound of this inequality (\ref{eq:-17}) is given by
\begin{equation}
\text{H}\left(\bar{C}\right)\leq\text{I}\left(A_{\left[n+1\right]};C\right)\text{\,\,\,\,[from \ensuremath{\bar{C}\leq C^{\left(0\right)}}]}\text{.}\label{eq:-18}
\end{equation}
We look for an upper bound on $\text{codim}_{B_{\left[n+1\right]}}B_{\left[n+1\right]}'$
in order to get a lower bound on $\text{H}\left(B_{\left[n+1\right]}'\right)$.
\[
\text{codim}_{B_{\left[n+1\right]}}B_{\left[n+1\right]}'\leq\stackrel[i=1]{n+1}{\sum}\text{codim}_{B_{i}}B_{i}'\text{\,\,\,\,[from lemma \ref{lemma basico 1desigualdad condicional rango lineal basica general}].}
\]
For $k\in\left[n\right]$, we have
\[
\text{codim}_{B_{k}}B_{k}'\leq\text{H}\left(B_{k}\mid A_{\left[n+1\right]-k}'\right)+\text{H}\left(B_{k}\mid A_{k}',C'\right)
\]
\[
=\text{codim}_{B_{k}}\left(A_{\left[n+1\right]-k}'\cap B_{k}\right)+\text{codim}_{B_{k}}\left(\left[A_{k}'+C'\right]\cap B_{k}\right)
\]
\[
=\text{codim}_{B_{k}}\left(A_{\left[n+1\right]-k}\cap B_{k}\right)+\text{codim}_{B_{k}}\left(\left[A_{k}+C\right]\cap B_{k}\right)
\]
\[
+\text{codim}_{A_{\left[n+1\right]-k}\cap B_{k}}\left(A_{\left[n+1\right]-k}'\cap B_{k}\right)+\text{codim}_{\left[A_{k}+C\right]\cap B_{k}}\left(\left[A_{k}'+C'\right]\cap B_{k}\right)
\]
\[
\leq\text{codim}_{B_{k}}\left(A_{\left[n+1\right]-k}\cap B_{k}\right)+\text{codim}_{B_{k}}\left(\left[A_{k}+C\right]\cap B_{k}\right)+\text{codim}_{A_{\left[n+1\right]-k}}\left(A_{\left[n+1\right]-k}'\right)
\]
\[
+\text{codim}_{A_{k}+C}\left(A_{k}'+C'\right)\,\,\,\,\,\text{\,\,\,\,[from lemma \ref{lema basico 2 }].}
\]
\[
=\text{H}\left(B_{k}\mid A_{\left[n+1\right]-k}\right)+\text{H}\left(B_{k}\mid A_{k},C\right)+\text{codim}_{A_{\left[n+1\right]-k}}\left(A_{\left[n+1\right]-k}'\right)+\text{codim}_{A_{k}+C}\left(A_{k}'+C'\right)
\]
\[
\leq\text{H}\left(B_{k}\mid A_{\left[n+1\right]-k}\right)+\text{H}\left(B_{k}\mid A_{k},C\right)+\stackrel[i=1]{n+1}{\sum}\text{codim}_{A_{i}}\left(A_{i}^{'}\right)+\text{codim}_{C}\left(\bar{C}\right)\text{\,\,\,\,[from lemma \ref{lemma basico 1desigualdad condicional rango lineal basica general}].}
\]
\[
\leq\text{H}\left(B_{k}\mid A_{\left[n+1\right]-k}\right)+\text{H}\left(B_{k}\mid A_{k},C\right)+\stackrel[i=2]{n+1}{\sum}\text{I}\left(A_{\left[i-1\right]};A_{i}\right)+\text{H}\left(C\mid A_{\left[n+1\right]}\right)+\stackrel[i=1]{n+1}{\sum}\text{I}\left(A_{\left[n+1\right]-i};C\right)\text{[from \ref{ecuaciones de codimension de los Ai primas}\text{]}.}
\]
For $k=n+1$, noting that $\text{codim}_{A_{\left[n\right]}}A_{\left[n\right]}'=0$,
we get
\[
\text{codim}_{B_{n+1}}B_{n+1}'\leq\text{H}\left(B_{n+1}\mid A_{\left[n\right]}\right)+\text{H}\left(B_{n+1}\mid A_{n+1},C\right)+\text{I}\left(A_{\left[n\right]};A_{n+1}\right)+\text{H}\left(C\mid A_{\left[n+1\right]}\right)+\stackrel[i=1]{n+1}{\sum}\text{I}\left(A_{\left[n+1\right]-i};C\right)\text{.}
\]
Then, we find that
\[
\text{codim}_{B_{\left[n+1\right]}}B_{\left[n+1\right]}'\leq\stackrel[i=1]{n+1}{\sum}\text{H}\left(B_{i}\mid A_{\left[n+1\right]-i}\right)+\stackrel[i=1]{n+1}{\sum}\text{H}\left(B_{i}\mid A_{i},C\right)+\stackrel[i=1]{n}{\sum}\text{I}\left(A_{\left[i\right]};A_{\left[n+1\right]-\left[i\right]}\right)
\]
\begin{equation}
+\stackrel[i=2]{n+1}{\sum}\text{I}\left(A_{\left[i-1\right]};A_{i}\right)+\left(n+1\right)\left[\text{H}\left(C\mid A_{\left[n+1\right]}\right)+\stackrel[i=1]{n+1}{\sum}\text{I}\left(A_{\left[n+1\right]-i};C\right)\right]\text{.}\label{eq: teorema principal en caracteristica que divide a n conta inferior sobre los B primas}
\end{equation}
From (\ref{eq:-17}) , (\ref{eq:-18}) and (\ref{eq: teorema principal en caracteristica que divide a n conta inferior sobre los B primas}),
we get the desired inequality. The inequality does not hold in general
over vector spaces whose characteristic does not divide $n$. A counter
example would be: In $V=\text{GF}\left(p\right)^{n+1}$, $p\nmid n$,
take the vector space $A_{1}$, $\ldots$, $A_{n+1}$, $B_{1}$, $\ldots$,
$B_{n+1}$ and $C$ generated by the columns of the matrix $L_{n}$.
Then, all information measures are zero but $\text{H}\left(B_{\left[n+1\right]}\right)=n+1$
and $\text{I}\left(A_{\left[n+1\right]};C\right)=1$. We get $n\geq n+1$
which is a contradiction.
\end{proof}
\begin{prop}
\label{prop:caso cuando la desigualdad fano es valida para todo}If
the dimension of vector space $V$ is at most $n$, then inequality
implicated by Theorem \ref{theorema fuerte caracterisitica rango lineal para conjunto finito}
is true over any field.
\end{prop}

\begin{proof}
We suposse that there exist vector subspaces $A_{1}$, $A_{2}$, $\ldots$,
$A_{n+1}$, $B_{1}$, $B_{2}$, $\ldots$, $B_{n+1}$, $C$ of a vector
space $V$ of dimension at most $n$ that do not hold the desired
inequality i.e. 
\[
\text{H}\left(B_{\left[n+1\right]}\right)>n\text{I}\left(A_{\left[n+1\right]};C\right)+\stackrel[i=1]{n+1}{\sum}\text{H}\left(B_{i}\mid A_{\left[n+1\right]-i}\right)+\stackrel[i=1]{n+1}{\sum}\text{H}\left(B_{i}\mid A_{i},C\right)+n\stackrel[i=2]{n}{\sum}\text{I}\left(A_{\left[i-1\right]};A_{i}\right)
\]
\[
+\left(n+1\right)\left[\text{I}\left(A_{\left[n\right]};A_{n+1}\right)+\text{H}\left(C\mid A_{\left[n+1\right]}\right)+\stackrel[i=1]{n+1}{\sum}\text{I}\left(A_{\left[n+1\right]-i};C\right)\right],
\]
and find a contradiction. Since $\text{H}\left(B_{\left[n+1\right]}\right)\leq n$,
the right side of the inequality is at most $n-1$. Hence, $\text{I}\left(A_{\left[i-1\right]};A_{i}\right)=\text{I}\left(A_{\left[n+1\right]};C\right)=\text{H}\left(C\mid A_{\left[n+1\right]}\right)=0$
for all $i$. So, we get $\bigoplus A_{i}$ is a direct sum and $C=O$
are the zero space. Then, the inequality becames $\text{H}\left(B_{\left[n+1\right]}\right)>\stackrel[i=1]{n+1}{\sum}\left[\text{H}\left(B_{i}\mid A_{\left[n+1\right]-i}\right)+\text{H}\left(B_{i}\mid A_{i}\right)\right]$.
We note that if $\text{H}\left(B_{i}\mid A_{\left[n+1\right]-i}\right)=0$
then $\text{H}\left(B_{i}\mid A_{i}\right)=\text{H}\left(B_{i}\right)$;
if $\text{H}\left(B_{i}\mid A_{i}\right)=0$ then $\text{H}\left(B_{i}\mid A_{\left[n+1\right]-i}\right)=\text{H}\left(B_{i}\right)$,
and at least $n+3$ summands are zeros in the right side of the inequality.
With this in mind, we get an inequality of the form $\text{H}\left(B_{S}\right)>\underset{i\in S}{\sum}\text{H}\left(B_{i}\right)$,
where $B_{i}\neq O$ for $i\in S$ which is a contradiction.
\end{proof}
We want to remark that the characteristic-dependent linear rank inequalities
in \cite{4}, which is valid for fields whose characteristic is different
from two, has an error which is produced by a failure in determining
an upper bound on the rank of a vector space in the demonstration
of \cite[ Theorem 6.2]{4}. A counter example for that inequality
would be: Let $V_{100}=\left\langle \left(\begin{array}{c}
1\\
0\\
0
\end{array}\right)\right\rangle $, $V_{010}=\left\langle \left(\begin{array}{c}
0\\
1\\
0
\end{array}\right)\right\rangle $, $V_{001}=\left\langle \left(\begin{array}{c}
0\\
0\\
1
\end{array}\right)\right\rangle $, $V_{011}=V_{101}=V_{110}=O$ and $V_{111}=\left\langle \left(\begin{array}{c}
1\\
1\\
1
\end{array}\right)\right\rangle $ be vector subspace of $\text{GF}\left(p\right)^{3}$ with $p\neq2$.
Then we get $-3\geq0$ which is a contradiction. So, in the case $n=2$,
the following inequality corrects this error.
\begin{thm}
\label{teorema fuerte desigualdad caracteristica rango lineal para cofinitos}For
any $n\geq2$. Let $A_{1}$, $A_{2}$, $\ldots$, $A_{n+1}$, $B_{1}$,
$B_{2}$, $\ldots$, $B_{n+1}$, $C$ be subspaces of a finite-dimensional
vector space $V$ over a scalar field $\mathbb{F}$ whose field characteristic
does not divide $n$,
\[
\mathrm{H}\left(C\right)\leq\frac{1}{n+1}\mathrm{H}\left(B_{\left[n+1\right]}\right)+\mathrm{H}\left(C\mid A_{\left[n+1\right]}\right)+\stackrel[i=1]{n+1}{\sum}\mathrm{I}\left(A_{\left[n+1\right]-i};C\right)+\stackrel[i=1]{n+1}{\sum}\mathrm{H}\left(C\mid A_{i},B_{i}\right)
\]
\[
+n\stackrel[i=2]{n}{\sum}\mathrm{I}\left(A_{\left[i-1\right]};A_{i}\right)+\left(n+1\right)\mathrm{I}\left(A_{\left[n\right]};A_{n+1}\right)+\stackrel[i=1]{n+1}{\sum}\mathrm{H}\left(B_{i}\mid A_{\left[n+1\right]-i}\right)\text{.}
\]
\end{thm}

\begin{proof}
The tuple (\ref{eq: construcci=0000F3n de espacios que son sumas directas})
obtained from the given vector spaces satisfies the condition (i)
of the lemma \ref{lema desigualdad condicional sobre caracteristica que no dividen n}
in the space $V'=A_{\left[n+1\right]}$. To meet condition (ii), we
define for \textbf{$k=1$ }to\textbf{ $k=n+1$}, $B_{k}':=B_{k}\cap\left(A_{\left[n+1\right]-k}'\right)$.
We get
\begin{equation}
\text{codim}_{B_{k}}\left(B_{k}'\right)\leq\text{H}\left(B_{k}\mid A_{\left[n+1\right]-k}\right)+\stackrel[i=2,i\neq k]{n+1}{\sum}\text{I}\left(A_{\left[i-1\right]};A_{i}\right),\text{ \ensuremath{k\in\left[n\right]}}\label{eq:-3-1}
\end{equation}
\begin{equation}
\text{codim}_{B_{n+1}}\left(B_{n+1}'\right)=\text{H}\left(B_{n+1}\mid A_{\left[n\right]}\right).\label{eq:-1-1}
\end{equation}
By last, to meet condition (iii), we obtain a new subspace of $\bar{C}$
that also satisfies (i) by following way. Let $\bar{C}^{\left(0\right)}:=\bar{C}$,
for $k=1$ to $k=n+1$, denote by $\bar{C}^{\left(k\right)}:=\bar{C}^{\left(k-1\right)}\cap\left(A_{k}'+B_{k}'\right)$.
Define $\hat{C}=\bar{C}^{\left(n+1\right)}$. The subspaces $A_{1}'$,
$...$, $A_{n+1}'$, $B_{1}'$, $...$, $B_{n+1}'$, $\hat{C}$ of
$V'$ satisfy all hypothesis of lemma \ref{lema desigualdad condicional sobre caracteristica que no dividen n},
we get
\begin{equation}
\left(n+1\right)\text{H}\left(\hat{C}\right)\leq\text{H}\left(B_{\left[n+1\right]}'\right)\label{eq:Caso distinto de 3 desigualdad condicional basica-1}
\end{equation}
We have to get an upper bound and a lower bound using (\ref{eq:Caso distinto de 3 desigualdad condicional basica-1}).
Obviously,
\begin{equation}
\text{H}\left(B_{\left[n+1\right]}'\right)\leq\text{H}\left(B_{\left[n+1\right]}\right).\label{eq:-20}
\end{equation}
We look for an upper bound on $\text{codim}_{C}\hat{C}$ in order
to get a lower bound on $\text{H}\left(\hat{C}\right)$,
\[
\text{codim}_{C}\hat{C}=\text{codim}_{C}\bar{C}+\text{codim}_{\bar{C}}\hat{C}
\]
\[
\,\,\,\,\,\,\,\,\,\,\,\,\,\,\,\,\,\,\,\,\,\,\,\,\,\,\,\,\leq\text{H}\left(C\mid A_{\left[n+1\right]}\right)+\stackrel[i=1]{n+1}{\sum}\text{I}\left(A_{\left[n+1\right]-i};C\right)+\stackrel[k=1]{n+1}{\sum}\text{H}\left(\bar{C}\mid A_{k}'+B_{k}'\right)\text{ [from (\ref{eq: codimension de C (raya horizontal arriba)}) and definition of \ensuremath{\hat{C}}]}
\]
\[
=\text{H}\left(C\mid A_{\left[n+1\right]}\right)+\stackrel[i=1]{n+1}{\sum}\text{I}\left(A_{\left[n+1\right]-i};C\right)+\stackrel[k=1]{n+1}{\sum}\text{codim}_{C}\left(C\cap\left[A_{k}'+B_{k}'\right]\right)
\]
\[
=\text{H}\left(C\mid A_{\left[n+1\right]}\right)+\stackrel[i=1]{n+1}{\sum}\text{I}\left(A_{\left[n+1\right]-i};C\right)+\stackrel[k=1]{n+1}{\sum}\text{codim}_{C}\left(C\cap\left[A_{k}+B_{k}\right]\right)
\]
\[
+\stackrel[k=1]{n+1}{\sum}\text{codim}_{\left(C\cap\left[A_{k}+B_{k}\right]\right)}\left(C\cap\left[A_{k}'+B_{k}'\right]\right)
\]
\[
\leq\text{H}\left(C\mid A_{\left[n+1\right]}\right)+\stackrel[i=1]{n+1}{\sum}\text{I}\left(A_{\left[n+1\right]-i};C\right)+\stackrel[k=1]{n+1}{\sum}\text{H}\left(C\mid A_{k},B_{k}\right)
\]
\[
+\stackrel[k=1]{n+1}{\sum}\text{codim}_{\left(A_{k}+B_{k}\right)}\left(A_{k}'+B_{k}'\right)\text{\,\,\,\,\,\,\,\,\,\,\,[from lemma \ref{lema basico 2 } and (\ref{lemma basico 1desigualdad condicional rango lineal basica general})]}
\]
\[
\leq\text{H}\left(C\mid A_{\left[n+1\right]}\right)+\stackrel[i=1]{n+1}{\sum}\text{I}\left(A_{\left[n+1\right]-i};C\right)+\stackrel[k=1]{n+1}{\sum}\text{H}\left(C\mid A_{k},B_{k}\right)
\]
\[
+\stackrel[i=1]{n+1}{\sum}\text{codim}_{A_{i}}\left(A_{i}'\right)+\stackrel[i=1]{n+1}{\sum}\text{codim}_{B_{i}}\left(B_{i}'\right)
\]
\[
\leq\text{H}\left(C\mid A_{\left[n+1\right]}\right)+\stackrel[i=1]{n+1}{\sum}\text{I}\left(A_{\left[n+1\right]-i};C\right)+\stackrel[k=1]{n+1}{\sum}\text{H}\left(C\mid A_{k},B_{k}\right)
\]
\[
+\stackrel[i=1]{n+1}{\sum}\text{I}\left(A_{\left[i-1\right]};A_{i}\right)+\stackrel[i=1]{n}{\sum}\stackrel[j=1,j\neq i]{n+1}{\sum}\text{I}\left(A_{\left[j-1\right]};A_{j}\right)+\stackrel[i=1]{n+1}{\sum}\text{H}\left(B_{i}\mid A_{\left[n+1\right]-i}\right)\,\,\,\,\text{[from (\ref{eq:-3-1})]}
\]
\[
\leq\text{H}\left(C\mid A_{\left[n+1\right]}\right)+\stackrel[i=1]{n+1}{\sum}\text{I}\left(A_{\left[n+1\right]-i};C\right)+\stackrel[k=1]{n+1}{\sum}\text{H}\left(C\mid A_{k},B_{k}\right)
\]
\[
+n\stackrel[i=2]{n}{\sum}\text{I}\left(A_{\left[i-1\right]};A_{i}\right)+\left(n+1\right)\text{I}\left(A_{\left[n\right]};A_{n+1}\right)+\stackrel[i=1]{n+1}{\sum}\text{H}\left(B_{i}\mid A_{\left[n+1\right]-i}\right)\,\,\,\,\text{[from (\ref{eq:-3-1})]}
\]

From (\ref{eq:Caso distinto de 3 desigualdad condicional basica-1})
, (\ref{eq:-20}) and last inequality, we get the desired inequality.
The inequality does not hold in general over vector spaces whose characteristic
divides $n$. A counter example would be: In $V=\text{GF}\left(p\right)^{n+1}$,
$p\nmid n$, take the vector space $A_{1}$, $\ldots$, $A_{n+1}$,
$B_{1}$, $\ldots$, $B_{n+1}$ and $C$ generated by the columns
of the matrix $L_{n}$.. Then, all information measures are zero but
$\text{H}\left(B_{\left[n+1\right]}\right)=n$ and $\text{H}\left(C\right)=1$.
We get $n+1\leq n$ which is a contradiction.
\end{proof}
\begin{prop}
\label{prop:caso donde la desi no fano es valida para todo}If the
dimension of vector space $V$ is at most $n$, then inequality implicated
by Theorem \ref{teorema fuerte desigualdad caracteristica rango lineal para cofinitos}
is true over any field.
\end{prop}

\begin{proof}
We suppose that there exist vector subspaces $A_{1}$, $A_{2}$, $\ldots$,
$A_{n+1}$, $B_{1}$, $B_{2}$, $\ldots$, $B_{n+1}$, $C$ of a vector
space $V$ of dimension at most $n$ that do not hold the desired
inequality i.e. 
\[
\text{H}\left(C\right)>\frac{1}{n+1}\text{H}\left(B_{\left[n+1\right]}\right)+\text{H}\left(C\mid A_{\left[n+1\right]}\right)+\stackrel[i=1]{n+1}{\sum}\text{I}\left(A_{\left[n+1\right]-i};C\right)+\stackrel[i=1]{n+1}{\sum}\text{H}\left(C\mid A_{i},B_{i}\right)
\]
\[
+n\stackrel[i=2]{n}{\sum}\text{I}\left(A_{\left[i-1\right]};A_{i}\right)+\left(n+1\right)\text{I}\left(A_{\left[n\right]};A_{n+1}\right)+\stackrel[i=1]{n+1}{\sum}\text{H}\left(B_{i}\mid A_{\left[n+1\right]-i}\right),
\]
and find a contradiction. Since $\text{H}\left(A_{\left[n+1\right]}\right)\leq n$,
there exists at least one $A_{k}$ such that $A_{k}\leq A_{\left[n+1\right]-k}$.
Then, the summing $\text{H}\left(C\mid A_{\left[n+1\right]}\right)+\stackrel[i=1]{n+1}{\sum}\text{I}\left(A_{\left[n+1\right]-i};C\right)$
on the right side of the desired inequality can be write as $\text{H}\left(C\right)+\stackrel[i=1,i\neq k]{n+1}{\sum}\text{I}\left(A_{\left[n+1\right]-i};C\right)$.
Hence, the right side of the inequality has negative information measures
which is a contradiction.
\end{proof}
\begin{rem}
An alternative demonstration of the above proposition and Proposition
\ref{prop:caso cuando la desigualdad fano es valida para todo} can
be obtained by noting that Lemma \ref{lemma 1:  distingue caracteristicas que dividen n y la que no lo divide}
is trivial when the dimension of $V$ is at most $n$.
\end{rem}

\section{Network Coding}

We will first briefly review some concepts of network coding in order
to fix some index coding terms. We study network coding with networks
in representation circuit, see \cite{16}, so each node represents
a coding function and hence the same message flows every edge coming
out of the same node. We emphasize that this approach loses no generality
and can be modified to coincide with other network models such as
the one used by Dougherty et al. \cite{5,6}. Formally, a \emph{network}
$\mathcal{N}=\left(V,E\right)$ is an aciclic multidirected-graph.
There exist source and receiver nodes and a (demand) function $\tau$
from collection of receivers $T$ onto collection of sources $S$.
There exist an alphabet $\mathcal{A}$, and a finite collection of
$k$-tuples of $\mathcal{A}$ called messages. Each source node has
a message. A $\left(k,n\right)$-network code specifies a alphabet
$\mathcal{A}$, two natural numbers $k$ and $n$, and a collection
of functions, one for each node of the network $\left(f_{v}\right)_{v\in V}\text{,}$
such that 

- If $v$ is a source, $f_{v}=\text{id}_{\mathcal{A}^{k}}$ (these
functions are generally omitted).

- If $v$ is not neither source or receiver, $f_{v}$ is a function
from $\text{Im}f_{v^{-}}$ to $\mathcal{A}^{n}$, where $\text{Im}f_{v^{-}}:=\underset{w\in v^{-}}{\prod}\text{Im}f_{w}$.

- If $v$ is a receiver, $f_{v}$ is a function called decoding function
from $\text{Im}f_{v^{-}}$ to $\mathcal{A}^{k}$..........

A network code is \emph{linear} if all their functions are linear
fictions over the same finite field. 

To capture the idea of transmit information through the network, there
is another collection of functions $\left(f_{v}^{*}\right)_{v\in V}$
on $\mathcal{A}^{k\left|S\right|}$, specified by the network code,
defined by 

- $f_{v}^{*}:=\pi_{v}$, if $v$ is a source.

- $f_{v}^{*}\left(x\right):=f_{v}\left(f_{v^{-}}^{*}\left(x\right)\right)=f_{v}\left(\left(f_{w}^{*}\left(x\right)\right)_{w\in v^{-}}\right)$
for all $x\in\mathcal{A}^{k\left|S\right|}$, if $v\in V-S$.

The value $f_{v}^{*}\left(x\right)$ gives the message that is carried
on the node for a given tuple of messages $x$. A network code is
a\emph{ solution} if for all tuple of messages $x$ and $t\in T$,
$f_{t}^{*}\left(x\right)=x_{\tau\left(t\right)}$ (i.e. the demand
of each receiver is satisfied).

The network coding problem of $\mathcal{N}$ is to find some alphabet,
and efficient solution over this alphabet. The efficiency is measured
by the ratio $\frac{k}{n}$. The \emph{capacity of $\mathcal{N}$
respect to a class of functions $\mathcal{D}$ over $\mathcal{A}$}
is 
\[
\mathrm{C}_{\mathcal{D}}^{\mathcal{A}}\left(\mathcal{N}\right):=\sup\left\{ \frac{k}{n}:\text{ \ensuremath{\exists} a \ensuremath{\left(k,n\right)}-solution in }\mathcal{D}\text{ over }\mathcal{A}\right\} \text{.}
\]
$\mathcal{D}$ is usually though as the collection of all network
codes, in this case the capacity is usually refereed as \emph{non-linear
coding capacity}. Also $\mathcal{D}$ can be taken as the collection
of linear codes over determined finite fields (or over any finite
field).

A network is defined to be \cite{6,7}:
\begin{itemize}
\item \emph{Solvable over $\mathcal{A}$} if there exists a $\left(1,1\right)$-solution
over $\mathcal{A}$, and\emph{ solvable} if the network is solvable
over some $\mathcal{A}$.
\item \emph{Scalar linearly solvable over $\mathbb{F}$} if there exists
a $\left(1,1\right)$-linear solution over $\mathbb{F}$, and\emph{
scalar linearly solvable} if the network is scalar linearly solvable
over some $\mathbb{F}$.
\item \emph{(Vector) Linearly solvable over $\mathbb{F}$} if there exists
a $\ensuremath{\left(k,k\right)}$-linear solution over $\mathbb{F}$,
for some $\ensuremath{k}\geq1$,\emph{ and linearly solvable} if the
network is (vector) linearly solvable over some $\mathbb{F}$.
\item \emph{Asymptotically solvable over $\mathcal{A}$} if for any $\epsilon>0$,
there exists a $\ensuremath{\left(k,n\right)}$-solution over $\mathcal{A}$
such that $\frac{k}{n}>1-\epsilon$, and the network is \emph{asymptotically
solvable} if the network is asymptotically solvable over some $\mathcal{A}$.
\item \emph{Asymptotically linearly solvable over $\mathbb{F}$} if for
any $\epsilon>0$, there exists a $\ensuremath{\left(k,n\right)}$-
linear solution over $\mathbb{F}$ such that $\frac{k}{n}>1-\epsilon$,
and the network is \emph{asymptotically linearly solvable} if the
network is asymptotically linearly solvable over some $\mathbb{F}$.
\end{itemize}
In this paper, we will use the following class of networks.
\begin{defn}
Let $m$ be a natural number. A \emph{$m$-index coding-network} is
a network with sources $S$ and receivers $T$ and a collection $\left[m\right]$
of $m$-intermediate nodes called $m$-block such that $S\times\left[m\right]$,$\left[m\right]\times T$$\subseteq E$.
\end{defn}

The network in case $m=1$ is simply called index coding-network and
corresponds to the index coding instance studied in \cite{2,3}. In
this case, the set of messages indexed by nodes of $t^{-}\cap S$
is known as the additional information of $t$. The message carried
on intermediate node is called broadcast message. Also, the network
is completely determined by $\left(S,E^{*}\right)$, where $E^{*}:=\left\{ \left(\tau\left(t\right),t^{-}\cap S\right)\in E:t\in T\right\} $.
To refer to these networks, we write $\mathcal{N}=\left(S,E^{*}\right)$.
From this, it is easy to obtain other $m$-index coding network $\mathcal{N}\left[m\right]=\left(S,E\right)$,
letting $E=\left(S\times\left[m\right]\right)\cup\left(\left[m\right]\times T\right)\cup E^{*}$.
The relationship between $\mathcal{N}$ and $\mathcal{N}\left[m\right]$
is established by the following lemma.
\begin{lem}
\label{lem:relaci=0000F3n entre un m index coding y index coding}Let
$m\in\mathbb{N}$. A $\left(k,n\right)$-solution of index coding-network
$\mathcal{N}$, implies a $\left(mk,n\right)$-solution of $\mathcal{N}\left[m\right]$.
Indeed, $\text{C}_{\mathcal{D}}\left(\mathcal{N}\left[m\right]\right)=m\text{C}_{\mathcal{D}}\left(\mathcal{N}\right)$,
where $\mathcal{D}$ can be the collection of all the codes or linear
codes.
\end{lem}

\textbf{From parameter of index coding instances to network coding
parameters.} The broadcast rate for an index coding instance is defined
in \cite{2}. This parameter coincides with the inverse multiplicative
of the capacity of the index coding network associated to the instance.
In the following we show some results from \cite{4} in our network
coding context. 

We use the following linear program problem \cite{4}: \emph{The (LP)
linear program with constraint matrix $A$} for an index coding-network
$\mathcal{N}$ is to determine min$\left(z_{\emptyset}\right)$ for
tuples of non-negative real numbers $\left(z_{Y}\right)_{Y\subseteq S}$
such that
\begin{lyxlist}{00.00.0000}
\item [{(i)}] $z_{S}=\left|S\right|$
\item [{(ii)}] $\forall Z\subseteq Y$ $z_{Y}-z_{Z}\leq\left|Y-\text{cl}\left(Z\right)\right|$
, where $\text{cl}\left(Z\right):=Z\cup\left\{ s\in S:\exists\left(s,Y\right)\in E^{*},Y\subseteq Z\right\} .$
\item [{(iii)}] $Az\geq0$.
\end{lyxlist}
Optimal solution is denoted by $\mathrm{b}_{A}\left(\mathcal{N}\right)$.
The inverse multiplicative of this value is denoted\footnote{in case $\mathrm{b}=0$, $\mathrm{B}=\infty$.}
by $\mathrm{B}_{\mathrm{A}}\left(\mathcal{N}\right)$. We remark that
conditions (i) and (ii) are associated to information flow of $\mathcal{N}$,
and condition (iii) enumerates a list $A$ of constraints correspond
to information inequalities or (characteristic-dependent) linear rank
inequalities. When $A$ enumerates the constraints correspond to information
inequalities, $\mathrm{B}_{\mathrm{A}}$ is an upper bound on the
capacity of $\mathcal{N}$; when $A$ enumerates the constraints correspond
to (characteristic-dependent) linear rank inequalities, $\mathrm{B}_{\mathrm{A}}$
is an upper bound on the linear capacity of $\mathcal{N}$ over the
alphabets in which the linear rank inequalities are valid. This is
easy to see, consider a $\left(k,n\right)$-solution of $\mathcal{N}$
over $\mathcal{A}$. Let $X_{1}$, $\ldots$, $X_{\left|S\right|}$
be independent uniformly distributed random variables (associated
to messages) over $\mathcal{A}^{k}$ and $P$ be a random variable
(associated to broadcast message) over $\mathcal{A}^{n}$. Take the
base of the entropy function as $\left|\mathcal{A}\right|^{k}$. Let
$z_{Y}=\text{H}\left(X_{Y}\cup P\right)$, we can verify that $\left(z_{Y}\right)_{Y\subseteq S}$
is a feasible primal solution of linear program problem. Thus, $z_{\emptyset}\leq\text{H}\left(P\right)\leq\frac{n}{k}$,
yielding $\mathrm{C}\left(\mathcal{N}\right)\leq\mathrm{B}_{A}\left(\mathcal{N}\right)$.
The upper bound on the linear capacity is obtained in a similar way.
The subscript in $\mathrm{b}_{A}\left(\mathcal{N}\right)$ is omitted
when $A$ corresponds to the constraints of the submodular inequality.

The \emph{lexicographic product} of index coding networks $\mathcal{N}_{1}$
and $\mathcal{N}_{2}$, denoted by $\mathcal{N}_{1}\bullet\mathcal{N}_{2}$,
is a index coding network whose source set is $S_{1}\times S_{2}$.
Each receiver $t$ is indexed by a pair $\left(t_{1},t_{2}\right)$
of receivers of $\mathcal{N}_{1}$ and $\mathcal{N}_{2}$ such that
$\tau\left(t\right)=\left(\tau\left(t_{1}\right),\tau\left(t_{2}\right)\right)$
and $t^{-}\cap\left(S_{1}\times S_{2}\right)=\left[\left(t_{1}^{-}\cap S_{1}\right)\times S_{2}\right]\cup\left[\tau\left(t_{1}\right)\times\left(t_{2}^{-}\cap S_{2}\right)\right]$.
The $k$-fold lexicographic power of $\mathcal{N}$ is denoted by
$\mathcal{N}^{\bullet k}$. Since the broadcast rate is sub-multiplicative
and $\mathrm{b}$ is super-multiplicative under the lexicographic
products \cite{4}, the capacity of index coding-networks is super-multiplicative
and $\mathrm{B}$ is sub-multiplicative under the lexicographic products
i.e. $\mathrm{C}\left(\mathcal{N}_{1}\right)\mathrm{C}\left(\mathcal{N}_{2}\right)\leq\mathrm{C}\left(\mathcal{N}_{1}\bullet\mathcal{N}_{2}\right)$
and $\mathrm{B}\left(\mathcal{N}_{1}\bullet\mathcal{N}_{2}\right)\leq\mathrm{B}\left(\mathcal{N}_{1}\right)\mathrm{B}\left(\mathcal{N}_{2}\right)$.

We want to define linear programs, using our inequalities, whose solutions
behave super-multiplicatively under lexicographic products, we make
this by the following argument: In \cite[Theorem 6.3]{4}, it is presented
a matrix $B$ whose transpose matrix has the property that if $\alpha$
is the associated vector of a linear rank inequality over $\mathbb{F}$,
then $\beta=B^{t}\alpha$ is the associated vector of a tight linear
rank inequality\footnote{A linear inequality $\alpha\bullet v\geq0$ is called tight if it
is balanced and $\sum\alpha_{i}=1$.} over $\mathbb{F}$. We can take the associated vectors of the inequalities
of the Theorems \ref{theorema fuerte caracterisitica rango lineal para conjunto finito}
and \ref{teorema fuerte desigualdad caracteristica rango lineal para cofinitos}.
Then, we apply this matrix to get two tight characteristic-dependent
linear rank inequalities: For any $A_{1}$, $A_{2}$, $\ldots$, $A_{n+1}$,
$B_{1}$, $B_{2}$, $\ldots$, $B_{n+1}$, $C$ and $P$ vector subspaces
of $V$, we get
\[
\text{H}\left(B_{\left[n+1\right]}\mid P\right)+\stackrel[i=1]{n+1}{\sum}\text{H}\left(B_{i}\mid A_{\left[n+1\right]},B_{\left[n+1\right]-i},C,P\right)+\left(n+1\right)\text{H}\left(C\mid A_{\left[n+1\right]},B_{\left[n+1\right]},P\right)
\]
\[
\leq\left(n+1\right)\stackrel[i=1]{n+1}{\sum}\text{I}\left(A_{\left[n+1\right]-i};C\mid P\right)+n\text{I}\left(A_{\left[n+1\right]};C\mid P\right)+\stackrel[i=1]{n+1}{\sum}\text{H}\left(B_{i}\mid A_{\left[n+1\right]-i},P\right)
\]
\begin{equation}
+\stackrel[i=1]{n+1}{\sum}\text{H}\left(B_{i}\mid A_{i},C,P\right)+n\stackrel[i=2]{n}{\sum}\text{I}\left(A_{\left[i-1\right]};A_{i}\mid P\right)+\left(n+1\right)\text{I}\left(A_{\left[n\right]};A_{n+1}\mid P\right)+\left(n+1\right)\text{H}\left(C\mid A_{\left[n+1\right]},P\right)\label{eq:desigualdad en una variable adicional sobre divisores de n}
\end{equation}
when $\text{char}\left(\mathbb{F}\right)$ divides $n$;
\[
\text{H}\left(C\mid P\right)+\left(n+1\right)\text{H}\left(C\mid A_{\left[n+1\right]},B_{\left[n+1\right]},P\right)+\frac{n+2}{n+1}\stackrel[i=1]{n+1}{\sum}\text{H}\left(B_{i}\mid A_{\left[n+1\right]},B_{\left[n+1\right]-i},C,P\right)
\]
\[
\leq\frac{1}{n+1}\text{H}\left(B_{\left[n+1\right]}\mid P\right)+\text{H}\left(C\mid A_{\left[n+1\right]},P\right)+\stackrel[i=1]{n+1}{\sum}\text{I}\left(A_{\left[n+1\right]-i},;C\mid P\right)+\stackrel[i=1]{n+1}{\sum}\text{H}\left(C\mid A_{i},B_{i},P\right)
\]
\begin{equation}
+n\stackrel[i=2]{n}{\sum}\text{I}\left(A_{\left[i-1\right]};A_{i}\mid P\right)+\left(n+1\right)\text{I}\left(A_{\left[n\right]};A_{n+1}\mid P\right)+\stackrel[i=1]{n+1}{\sum}\text{H}\left(B_{i}\mid A_{\left[n+1\right]-i},P\right),\label{eq:desigualdad en una variable adicional sobre no divisores de n}
\end{equation}
when $\text{char}\left(\mathbb{F}\right)$ does not divide $n$. We
use these inequalities to define two new linear programs adding the
constraints imply by each one of theses inequalities to the matrix
$A$ of LP with constraint matrix given by submodular inequality.
The linear program which use the first inequality, we shall call LP-$\mathcal{A}_{n}$,
and the linear program which use the second inequality, we shall call
LP-$\mathcal{B}_{n}$. The optimal solutions are denoted by $\mathrm{b}_{\mathcal{A}_{n}}$
and $\mathrm{b}_{\mathcal{B}_{n}}$. The following inequality is a
constraint which is satisfied by  LP-$\mathcal{A}_{n}$, this is obtained
from inequality \ref{eq:desigualdad en una variable adicional sobre divisores de n}
and \cite[Lemma 6.4]{4},\smallskip{}
\[
\left(2n^{2}+3n+1\right)z_{\emptyset}+2\left(n+1\right)z_{A_{\left[n+1\right]},B_{\left[n+1\right]},C}+z_{B_{\left[n+1\right]}}+\stackrel[i=1]{n+1}{\sum}\left(z_{A_{i},C}+\left(n+1\right)z_{A_{\left[n+1\right]-i},C}\right)+\left(n+2\right)z_{A_{\left[n+1\right]}}
\]
\[
\leq+z_{A_{\left[n+1\right]},C}+n\stackrel[i=1]{n}{\sum}\left(z_{A_{i}}+z_{A_{\left[n+1\right]-i}}\right)+\left(n^{2}+3n+1\right)z_{C}
\]
\begin{equation}
+\left(n+1\right)\left(z_{A_{\left[n\right]}}+z_{A_{\left[n+1\right]},B_{\left[n+1\right]}}+z_{A_{n+1}}\right)+\stackrel[i=1]{n+1}{\sum}\left(z_{A_{\left[n+1\right]},B_{\left[n+1\right]-i},C}+z_{A_{\left[n+1\right]-i},B_{i}}+z_{A_{i},B_{i},C}\right);\label{eq:desigualdad del esquema tight y homo divisores de n}
\end{equation}
in analogous way, the following inequality is a constraint which is
satisfied by LP-$\mathcal{B}_{n}$, this is obtained from inequality
\ref{eq:desigualdad en una variable adicional sobre no divisores de n}
and \cite[Lemma 6.4]{4},
\[
+\left(2n+3\right)z_{A_{\left[n+1\right]},B_{\left[n+1\right]},C}+\stackrel[i=1]{n+1}{\sum}\left(z_{A_{\left[n+1\right]-i},C}+z_{A_{i},B_{i}}\right)+\left(n+2\right)z_{A_{\left[n+1\right]}}+\frac{n^{3}+2n^{2}+2n+2}{n+1}z_{\emptyset}
\]
\[
\leq\frac{1}{n+1}z_{B_{\left[n+1\right]}}+z_{C,A_{\left[n+1\right]}}+\left(n+1\right)z_{A_{\left[n+1\right]},B_{\left[n+1\right]}}+\frac{n+2}{n+1}\stackrel[i=1]{n+1}{\sum}z_{A_{\left[n+1\right]},B_{\left[n+1\right]-i},C}
\]
\begin{equation}
+z_{A_{\left[n\right]}}+n\stackrel[i=1]{n}{\sum}z_{A_{i}}+\left(n+1\right)z_{A_{n+1}}+nz_{C}+\stackrel[i=1]{n+1}{\sum}\left(z_{A_{i},B_{i},C}+z_{A_{\left[n+1\right]-i},B_{i}}\right).\label{eq:eq:desigualdad del esquema tight y homo NO divisores de n}
\end{equation}
By last, from \cite[Theorem 3.4]{4}, we get that optimal solutions
of our LP-problems are super-multiplicative under lexicographic products.\bigskip{}

\textbf{Index coding from matroids. }A matroid is an abstract structure
that captures the notion of independence in linear algebra \cite{15}.
Let $\mathcal{M}=\left(S,r\right)$ be a matroid and let $J$ be the
set of coloops of $\mathcal{M}$ (each element is in no circuit).
Consider the matroid obtained by deletion of $J$, $\mathcal{M}\mid J=\left(S-J,r\mid_{J}\right)$.
Define \emph{the index coding network associated to $\mathcal{M}$}
by an index coding-network, denoted by $\mathcal{N}_{\mathcal{M}}$,
 with source set $S-J$ and $E_{\mathcal{M}}^{*}:=\left\{ \left(s,C-s\right):C\text{ is a circuit in }\mathcal{M}\mid J,s\in C\right\} $.
This construction is a modification of the construction given by Blasiak
et al. \cite[Definition 5.1]{4}. Our network has a smaller number
of sources and receivers because it is completely determined by the
circuits of the matroid. We introduce the following definition in
order to study the properties of this network.\medskip{}

\begin{defn}
An index coding network $\mathcal{N}'=\left(S,E_{\mathcal{N}'}^{*}\right)$
is called an index coding-subnetwork of $\mathcal{N}$ if $E_{\mathcal{N}'}^{*}\subseteq E_{\mathcal{N}}^{*}$
and there exists a collection $\left\{ \left(s,S_{s}\right)\right\} _{s\in S}$
of elements of $E_{\mathcal{N}'}^{*}$ such that $T:=\bigcup_{s\in S}S_{s}$
is a minimum subset of $S$, with the property that for all $s\in S$,
$\left(s,T_{s}\right)\in E_{\mathcal{N}}^{*}$, for some $T_{s}\subseteq T$.
This is equivalent to $\text{cl}_{\mathcal{N}'}\leq\text{cl}_{\mathcal{N}}$
and $r_{\text{cl}_{\mathcal{N}'}}=r_{\text{cl}_{\mathcal{N}}}$, where
$r_{\text{cl}}:=\min\left\{ \left|T\right|:\text{cl}\left(T\right)=S\right\} $. 
\end{defn}

The definition of subnetwork guarantees that the network flow of a
subnetwork behaves like the network flow of the network. Specifically,
a solution of $\mathcal{N}$ is a solution of $\mathcal{N}'$ and
$\mathrm{b}\left(\mathcal{N}'\right)\leq\mathrm{b}\left(\mathcal{N}\right)$.
Furthermore, the index coding network of a matroid $\mathcal{M}$
is an index coding-subnetwork of the index coding-network obtained
from the index coding instance associated to the matroid $\mathcal{M}\mid J$
of Blasiak et al.. With this in mind, the following proposition (and
proof) is a rewriting of \cite[Proposition 5.2 and Theorem 5.4]{4}.
\begin{prop}
\label{proposici=0000F3n sobre capacidades de subredes}Let $\mathcal{M}=\left(S,r\right)$
be a matroid. For any index coding-subnetwork $\mathcal{N}$ of the
index coding network $\mathcal{N}_{\mathcal{M}}$,
\[
\mathrm{B}\left(\mathcal{N}\right)=\frac{1}{\left|S\right|-r_{\mathcal{M}}}.
\]
Also, if some $\mathcal{M}$ is representable over $\mathbb{F}$,
then 
\[
\mathrm{C}\left(\mathcal{N}\right)=\mathrm{C}_{\text{linear}}^{\mathbb{F}}\left(\mathcal{N}\right)=\frac{1}{\left|S\right|-r_{\mathcal{M}}}
\]
and this capacity is achieved by a $\left(1,\left|S\right|-r_{\mathcal{M}}\right)$-linear
solution over $\mathbb{F}$.
\end{prop}

\textbf{Applications. }We use index coding-networks from matroids
for our theorem. Fixed $n$. For a field $\mathbb{F}$, matrix $L_{n}$
over $\mathbb{F}$ induces a vector matroid $\mathcal{M}\left(L_{n}\right)$
with ground set $S:=\left\{ A_{1},\ldots,A_{n+1},B_{1},\ldots,B_{n+1},C\right\} $,
some of these are known in \cite{14} for $n$ prime. If we change
the field, it is possible that the vector matroid changed. However,
these matroids have some properties in common. Specifically, certain
subsets of the ground set of $\mathcal{M}\left(L_{n}\right)$ are
always circuits according to the characteristic of $\mathbb{F}$ divides
or does not $n$. We classify them in two types\footnote{Here we use the notation $A_{I}:=\left\{ A_{i}:i\in I\right\} $.}:
The collection $\mathcal{A}_{n}:=\left\{ A_{\left[n+1\right]}C,A_{\left[n+1\right]-i}B_{i},A_{i}B_{i}C,B_{\left[n+1\right]}:i\in\left[n+1\right]\right\} $
is a subclass of circuits in any $\mathcal{M}\left(L_{n}\right)$
over $\mathbb{F}$, when $\text{char}\left(\mathbb{F}\right)$ divides
$n$; and the collection $\mathcal{B}_{n}:=\left\{ A_{\left[n+1\right]}C,A_{\left[n+1\right]-i}B_{i},A_{i}B_{i}C,B_{\left[n+1\right]}C:i\in\left[n+1\right]\right\} $
is a subclass of circuits in any $\mathcal{M}\left(L_{n}\right)$
over $\mathbb{F}$, when $\text{char}\left(\mathbb{F}\right)$ does
not divide $n$. We define $\mathcal{N}_{\mathcal{A}_{n}}$ as the
index coding with the source set $S$ and $E_{\mathcal{A}_{n}}^{*}:=\left\{ \left(s,C-s\right):C\in\mathcal{A}_{n},s\in C\right\} $;
and $\mathcal{N}_{\mathcal{B}_{n}}$ as the index coding with the
source set $S$  and $E_{\mathcal{B}_{n}}^{*}:=\left\{ \left(s,C-s\right):C\in\mathcal{B}_{n},s\in C\right\} $.

Before continuing, the following statements are useful.
\begin{lem}
\label{lem:codigo linear a partir del producto lexi de index coding networks}For
any $\mathcal{N}_{1}$ and $\mathcal{N}_{2}$. If $\mathcal{N}_{1}$
has a $\left(n,m\right)$-linear solution and $\mathcal{N}_{2}$ has
a $\left(k,n\right)$-linear solution both over the same field, then
$\mathcal{N}_{1}\bullet\mathcal{N}_{2}$ has a $\left(k,m\right)$-linear
solution.
\end{lem}

\begin{proof}
Let $f$ be the function on the intermediate node and $f_{t_{1}}$
be the decoding function on a receiver $t_{1}$ of the desired $\left(n,m\right)$-linear
solution of $\mathcal{N}_{1}$, and let $g$ be the function on the
intermediate node and $g_{t_{2}}$ be the decoding function on a receiver
$t_{2}$ of the desired $\left(k,n\right)$-linear solution of $\mathcal{N}_{2}$.
Define $g'\left(x\right):=\left(g\left(x_{s_{1}\times S_{2}}\right)\right)_{s_{1}\in S_{1}}$,
$x\in\mathbb{F}^{k\left|S_{1}\times S_{2}\right|}$, and let $h=fg'$
be the function on the intermediate node in $\mathcal{N}_{1}\bullet\mathcal{N}_{2}$.
We obtain the broadcast message $h\left(x\right)\in\mathbb{F}^{m}$.
Let $t$ be a receiver in $\mathcal{N}_{1}\bullet\mathcal{N}_{2}$
such that $\tau\left(t\right)=\left(\tau\left(t_{1}\right),\tau\left(t_{2}\right)\right)$
and $t^{-}\cap\left(S_{1}\times S_{2}\right)=\left[\left(t_{1}^{-}\cap S_{1}\right)\times S_{2}\right]\cup\left[\tau\left(t_{1}\right)\times\left(t_{2}^{-}\cap S_{2}\right)\right]$.
We have $f_{t_{1}}\left(h\left(x\right),\left(g\left(x_{s_{1}\times S_{2}}\right)\right)_{s_{1}\in t_{1}^{-}\cap S_{1}}\right)=f_{t_{1}}\left(f\left(\left(g\left(x_{s_{1}\times S_{2}}\right)\right)_{s_{1}\in S_{1}}\right),\left(g\left(x_{s_{1}\times S_{2}}\right)\right)_{s_{1}\in t_{1}^{-}\cap S_{1}}\right)=g\left(x_{\tau\left(t_{1}\right)\times S_{2}}\right)$.
Then, $g_{t_{2}}\left(g\left(x_{\tau\left(t_{1}\right)\times S_{2}}\right),x_{\tau\left(t_{1}\right)\times\left(t_{2}^{-}\cap S_{2}\right)}\right)=x_{\left(\tau\left(t_{1}\right),\tau\left(t_{2}\right)\right)}$.
These equations and $h$ clearly define a $\left(k,m\right)$-linear
solution of $\mathcal{N}_{1}\bullet\mathcal{N}_{2}$.
\end{proof}
\begin{lem}
\emph{\label{prop: codigo lineal para una potencia lexicografica de una index coding network}}For
$k\in\mathbb{N}$. If $\mathcal{N}$ has a $\left(1,n\right)$-linear
solution, then $\mathcal{N}^{\bullet k}$ has a $\left(1,n^{k}\right)$-linear
solution. 
\end{lem}

\begin{proof}
By induction, case $k=2$, take $\mathcal{N}_{1}=\mathcal{N}_{2}=\mathcal{N}$
in Lemma \ref{lem:codigo linear a partir del producto lexi de index coding networks}
and note that $\mathcal{N}_{2}$ has a $\left(n,n^{2}\right)$-linear
solution by repetition of the given solution of $\mathcal{N}$. We
get a $\left(1,n^{2}\right)$-linear solution of $\mathcal{N}^{\bullet2}$.
Now, we suppose that case $k-1$ holds i.e. $\mathcal{N}^{\bullet k-1}$
has a $\left(1,n^{k-1}\right)$-linear solution. Take $\mathcal{N}_{1}=\mathcal{N}$,
$\mathcal{N}_{2}=\mathcal{N}^{\bullet k-1}$ in Lemma \ref{lem:codigo linear a partir del producto lexi de index coding networks}
and note that $\mathcal{N}_{1}$ has a $\left(n^{k-1},n^{k}\right)$-linear
solution by repetition of the given solution of $\mathcal{N}$. Then,
$\mathcal{N}^{\bullet k}$ has a $\left(1,n^{k}\right)$-linear solution.
\end{proof}
\begin{thm}
\label{Teorema separaci=0000F3n de redes via caracter=0000EDsticas lineal y lineal}For
any $k,n\in\mathbb{N}$, $n\geq2$. We have,

(i) $\mathcal{N}_{\mathcal{A}_{n}}^{\bullet k}\left[\left(n+2\right)^{k}\right]$
is linearly solvable over a field $\mathbb{F}$ if, and only if, $\text{char}\left(\mathbb{F}\right)$
divides $n$. Also, when $\text{char}\left(\mathbb{F}\right)\nmid n$,
\[
\left(\frac{n+2}{n+3}\right)^{k}\leq\mathrm{C}_{\text{linear}}^{\mathbb{F}}\left(\mathcal{N}_{\mathcal{A}_{n}}^{\bullet k}\left[\left(n+2\right)^{k}\right]\right)\leq\left(\frac{5n^{3}+22n^{2}+31n+14}{5n^{3}+22n^{2}+31n+15}\right)^{k}.
\]

(ii) $\mathcal{N}_{\mathcal{B}_{n}}^{\bullet k}\left[\left(n+2\right)^{k}\right]$
is linearly solvable over a field $\mathbb{F}$ if, and only if, $\text{char}\left(\mathbb{F}\right)$
does not divide $n$. Also, when $\text{char}\left(\mathbb{F}\right)\mid n$,
\[
\left(\frac{n+2}{n+3}\right)^{k}\leq\mathrm{C}_{\text{linear}}^{\mathbb{F}}\left(\mathcal{N}_{\mathcal{B}_{n}}^{\bullet k}\left[\left(n+2\right)^{k}\right]\right)\leq\left(\frac{n^{3}+8n^{2}+19n+14}{n^{3}+8n^{2}+19n+15}\right)^{k}.
\]
 
\end{thm}

\begin{proof}
For (i), we have that $\mathcal{N}_{\mathcal{A}_{n}}$ is an index
coding-subnetwork of any $\mathcal{N}_{\mathcal{M}\left(L_{n}\right)}$
when $\text{char}\left(\mathbb{F}\right)$ divides $n$. Using Lemma
\ref{proposici=0000F3n sobre capacidades de subredes}, we have $\mathrm{C}\left(\mathcal{N}_{\mathcal{A}_{n}}\right)=\mathrm{C}_{\text{linear}}^{\mathbb{F}}\left(\mathcal{N}_{\mathcal{A}_{n}}\right)=\frac{1}{n+2}$
when $\text{char}\left(\mathbb{F}\right)$ divides $n$ and this capacity
is achieved by a $\left(1,n+2\right)$-linear solution over $\mathbb{F}$.
By Lemma \ref{prop: codigo lineal para una potencia lexicografica de una index coding network}
with $\mathcal{N}=\mathcal{N}_{\mathcal{A}_{n}}$, $\mathcal{N}_{\mathcal{A}_{n}}^{\bullet k}$
has a $\left(1,\left(n+2\right)^{k}\right)$-linear solution over
$\mathbb{F}$. Finally, by Lemma \ref{lem:relaci=0000F3n entre un m index coding y index coding},
$\mathcal{N}_{\mathcal{A}_{n}}^{\bullet k}\left[\left(n+2\right)^{k}\right]$
has a $\left(\left(n+2\right)^{k},\left(n+2\right)^{k}\right)$-linear
solution over $\mathbb{F}$ which implies that $\mathcal{N}_{\mathcal{A}_{n}}^{\bullet k}\left[\left(n+2\right)^{k}\right]$
is linearly solvable over a field $\mathbb{F}$ whose $\text{char}\left(\mathbb{F}\right)$
divides $n$. We estimate an upper bound on $\mathrm{C}_{\text{linear}}^{\mathbb{F}}\left(\mathcal{N}_{\mathcal{A}_{n}}\right)$
when $\text{char}\left(\mathbb{F}\right)$ does not divide $n$, using
the LP-$\mathcal{B}_{n}$: Let $\left(z_{S}\right)_{S\subseteq V}$
be a solution of LP-$\mathcal{B}_{n}$ for $\mathcal{N}_{\mathcal{A}_{n}}$.
From definition of $\mathcal{N}_{\mathcal{A}_{n}}$, we have:

(a) If $Y$ is a dependent set in each $\mathcal{M}\left(L_{n}\right)$
($\text{char}\left(\mathbb{F}\right)$ divides $n$), then $z_{Y}\leq z_{\emptyset}+r_{\mathcal{M}\left(L_{n}\right)}\left(Y\right)$. 

(b) If $Y$ is an independent set in each $\mathcal{M}\left(L_{n}\right)$
($\text{char}\left(\mathbb{F}\right)$ divides $n$), then $\left|Y\right|+n+2\leq z_{Y}\leq\left|Y\right|+z_{\emptyset}$.

We can use constraints implied by these conditions along with the
constraint \ref{eq:eq:desigualdad del esquema tight y homo NO divisores de n}
to get $z_{\emptyset}\geq\frac{5n^{3}+22n^{2}+31n+15}{5n^{2}+12n+7}$
which implies that $\mathrm{b}_{\mathcal{B}_{n}}\left(\mathcal{N}_{\mathcal{A}_{n}}\right)\geq\frac{5n^{3}+22n^{2}+31n+15}{5n^{2}+12n+7}$.
By super-multiplicative of $\mathrm{b}_{\mathcal{B}_{n}}$ under lexicographic
products, $\mathrm{b}_{\mathcal{B}_{n}}\left(\mathcal{N}_{\mathcal{A}_{n}}^{\bullet k}\right)\geq\left(\frac{5n^{3}+22n^{2}+31n+15}{5n^{2}+12n+7}\right)^{k}$.
Then $\mathrm{C}_{\text{linear}}^{\mathbb{F}}\left(\mathcal{N}_{\mathcal{A}_{n}}^{\bullet k}\right)\leq\left(\frac{5n^{2}+12n+7}{5n^{3}+22n^{2}+31n+15}\right)^{k}$.
Hence, using Lemma \ref{lem:relaci=0000F3n entre un m index coding y index coding}
with $m=\left(n+2\right)^{k}$, $\mathrm{C}_{\text{linear}}^{\mathbb{F}}\left(\mathcal{N}_{\mathcal{A}_{n}}^{\bullet k}\left(\left(n+2\right)^{k}\right)\right)\leq\left(\frac{5n^{3}+22n^{2}+31n+14}{5n^{3}+22n^{2}+31n+15}\right)^{k}$
$<1$, when $\text{char}\left(\mathbb{F}\right)$ does not divide
$n$.

For (ii), we have that $\mathcal{N}_{\mathcal{B}_{n}}$ is an index
coding-subnetwork of any $\mathcal{N}_{\mathcal{M}\left(L_{n}\right)}$
when $\text{char}\left(\mathbb{F}\right)$ does not divide $n$. Using
Lemma \ref{proposici=0000F3n sobre capacidades de subredes}, we have
$\mathrm{C}\left(\mathcal{N}_{\mathcal{B}_{n}}\right)=\mathrm{C}_{\text{linear}}^{\mathbb{F}}\left(\mathcal{N}_{\mathcal{B}_{n}}\right)=\frac{1}{n+2}$
when $\text{char}\left(\mathbb{F}\right)$ does not divide $n$ and
this capacity is achieved by a $\left(1,n+2\right)$-linear solution
over $\mathbb{F}$. Then, we apply an argument as in (i) to get the
required linear solution of $\mathcal{N}_{\mathcal{B}_{n}}^{\bullet k}\left[\left(n+2\right)^{k}\right]$.
We estimate an upper bound on $\mathrm{C}_{\text{linear}}^{\mathbb{F}}\left(\mathcal{N}_{\mathcal{B}_{n}}\right)$
when $\text{char}\left(\mathbb{F}\right)$ divides $n$ using the
LP-$\mathcal{A}_{n}$: Let $\left(z_{S}\right)_{S\subseteq V}$ be
a solution of LP-$\mathcal{A}_{n}$ for $\mathcal{N}_{\mathcal{B}_{n}}$.
From definition of $\mathcal{N}_{\mathcal{B}_{n}}$, we have that
this network satisfies conditions (a)-(b) of part (i) when the matroid
$\mathcal{M}\left(L_{n}\right)$ is taken over a field $\mathbb{F}$
whose $\text{char}\left(\mathbb{F}\right)$ does not divide $n$.
We can use constraints implied by these conditions along with the
constraint \ref{eq:desigualdad del esquema tight y homo divisores de n}
to get $z_{\emptyset}\geq\frac{n^{3}+8n^{2}+19n+15}{n^{2}+6n+7}$
which implies that $\mathrm{b}_{\mathcal{A}_{n}}\left(\mathcal{N}_{\mathcal{B}_{n}}\right)\geq\frac{n^{3}+8n^{2}+19n+15}{n^{2}+6n+7}$.
Then, by super-multiplicative of $\mathrm{b}_{\mathcal{A}_{n}}$ under
lexicographic products, $\mathrm{b}_{\mathcal{A}_{n}}\left(\mathcal{N}_{\mathcal{B}_{n}}^{\bullet k}\right)\geq\left(\frac{n^{3}+8n^{2}+19n+15}{n^{2}+6n+7}\right)^{k}$.
Thus, $\mathrm{C}_{\text{linear}}^{\mathbb{F}}\left(\mathcal{N}_{\mathcal{B}_{n}}^{\bullet k}\right)\leq\left(\frac{n^{2}+6n+7}{n^{3}+8n^{2}+19n+15}\right)^{k}$,
when $\text{char}\left(\mathbb{F}\right)$ divide $n$. Hence, using
Lemma \ref{lem:relaci=0000F3n entre un m index coding y index coding}
with $m=\left(n+2\right)^{k}$, $\mathrm{C}_{\text{linear}}^{\mathbb{F}}\left(\mathcal{N}_{\mathcal{B}_{n}}^{\bullet k}\left(\left(n+2\right)^{k}\right)\right)\leq\left(\frac{n^{3}+8n^{2}+19n+14}{n^{3}+8n^{2}+19n+15}\right)^{k}<1$,
when $\text{char}\left(\mathbb{F}\right)$ divides $n$.

For the remaining lower bounds on the linear capacities over fields
in which the networks are not linearly solvable, we use the network
topology in common of $\mathcal{N}_{\mathcal{A}_{n}}$ and $\mathcal{N}_{\mathcal{B}_{n}}$:
We add the message of $C$ to the broadcast message of the $\left(1,n+2\right)$-linear
solution of $\mathcal{N}_{\mathcal{B}_{n}}$ over $\mathbb{F}$ when
$\text{char}\left(\mathbb{F}\right)$ does not divide $n$ to obtain
a $\left(1,n+3\right)$-linear code which is a linear solution of
$\mathcal{N}_{\mathcal{A}_{n}}$ over this field. Then, the solution
is extended to a $\left(\left(n+2\right)^{k},\left(n+3\right)^{k}\right)$-linear
solution of $\mathcal{N}_{\mathcal{A}_{n}}^{\bullet k}\left[\left(n+2\right)^{k}\right]$
yielding $\left(\frac{n+2}{n+3}\right)^{k}\leq\mathrm{C}_{\text{linear}}^{\mathbb{F}}\left(\mathcal{N}_{\mathcal{A}_{n}}^{\bullet k}\left[\left(n+2\right)^{k}\right]\right)$.
In an analogous way, we get the respective lower bound on $\mathrm{C}_{\text{linear}}^{\mathbb{F}}\left(\mathcal{N}_{\mathcal{B}_{n}}^{\bullet k}\left[\left(n+2\right)^{k}\right]\right)$,
when $\text{char}\left(\mathbb{F}\right)$ divides $n$.
\end{proof}
\begin{cor}
Let $P$ be a finite or co-finite set of primes. There exists a sequence
of networks $\left(\mathcal{N}_{P}^{k}\right)_{k}$ in which each
member is linearly solvable over a field $\mathbb{F}$ if and only
if the characteristic of $\mathbb{F}$ is in $P$. Furthermore, when
$\text{char}\left(\mathbb{F}\right)$ is not in $P$, $\mathrm{C}_{\text{linear}}^{\mathbb{F}}\left(\mathcal{N}_{P}^{k}\right)\rightarrow0$
as $k\rightarrow\infty$.
\end{cor}

\begin{proof}
In the previous theorem, take $n=\underset{p\in P}{\prod}P$ if $P$
is finite and $n=\underset{p\notin P}{\prod}P$ if $P$ is co-finite.
\end{proof}
The following corollary is a straightforward consequence of the theorem
\ref{Teorema separaci=0000F3n de redes via caracter=0000EDsticas lineal y lineal},
and it is a generalization of \cite[Theorem 1.2]{4}. The proof is
followed taking: $\mathcal{N}_{n}=\mathcal{N}_{\mathcal{A}_{n}}\bullet\mathcal{N}_{\mathcal{B}_{n}}$,
and for all $k\in\mathbb{N}$, $\mathcal{N}_{n}^{k}:=\mathcal{N}_{n}^{\bullet k}\left[\left(n+2\right)^{2k}\right]$.
Then, we apply an argument as the previous theorem.
\begin{cor}
\label{Teorema separaci=0000F3n de redes via lineal y no lineal}There
exists a infinite collection of sequences of networks $\left\{ \left(\mathcal{N}_{n}^{k}\right)_{k}:n\in\mathbb{N},n\geq2\right\} $
in which each member of each sequence is asymptotically solvable but
is not asymptotically linearly solvable and the linear capacity $\rightarrow0$
as $k\rightarrow\infty$ in each sequence.
\end{cor}

The network coding gain is equal to the coding capacity divided by
the routing capacity. In \cite{11,12}, there are two sequences of
networks $\mathcal{N}_{i}\left(k\right)$ ($i=1,2$) such that the
coding gain $\rightarrow\infty$ as $k\rightarrow\infty$. The routing
capacities of $\mathcal{N}_{P}^{k}$ and $\mathcal{N}_{n}^{k}$ are
$\left(\frac{n+2}{2n+3}\right)^{k}$ and $\left(\frac{n^{2}+2n+4}{4n^{2}+12n+9}\right)^{k},$
respectively. Hence, any sequence of networks presented previously
satisfies this property.
\begin{cor}
The network coding gain of the sequences $\left(\mathcal{N}_{P}^{k}\right)_{k}$
and $\left(\mathcal{N}_{n}^{k}\right)_{k}$ $\rightarrow\infty$ as
$k\rightarrow\infty$.
\end{cor}

\section*{Acknowledgments}

The first author thanks the support provided by COLCIENCIAS through
Convocatoria 727 and the second author thanks the support provided
by Universidad Nacional de Colombia through The Hermes Research System
project 37216.

\end{document}